\documentclass[12pt]{amsart}
\usepackage{amssymb,amsfonts,amsthm}
\usepackage{cite}
\usepackage{epsfig,color}
\usepackage{subfigure}
\newcommand{\cre}{\textcolor{black}}

\def\CwedgeX{\mbox{$\bigwedge^{2}({\mathbb C}^{4})$}}

\newcommand{\rd}{\mathrm{d}}            
\newcommand{\ri}{\mathrm{i}}            

\newcommand{\hy}{{\widehat{y}}}
\newcommand{\hw}{{\widehat{w}}}

\newcommand{\op}{{\overline{p}}}
\newcommand{\oq}{{\overline{q}}}
\newcommand{\Le}{{\rm{Le}}}
\newcommand{\eps}{\varepsilon}

\newcommand{\realpart}[1]{\operatorname{\sf Re}\left(#1\right)}
\newcommand{\impart}[1]{\operatorname{\sf Im}\left(#1\right)}
\newcommand{\ww}{{hu+(1-h)(1-y)}}

\newcommand{\eq}[2]{\begin{equation}\begin{split}#1\end{split}\label{#2}\end{equation}}

\newcommand{\eqnn}[1]{\begin{equation}\begin{split}#1\end{split}\nonumber\end{equation}}

\newcommand{\iR}[1]{\int_{\mathbb{R}} #1 \,{\mbox{d}}\xi}

\newtheorem{theorem}{Theorem}
\numberwithin{theorem}{section}

\newtheorem{lemma}[theorem]{Lemma}

\theoremstyle{definition}

\renewcommand\rightmark{Combustion  waves in hydraulically resistant porous media  \hfill\thepage}

\begin{document}
{\textcolor{red}{\title{\cre{Combustion  waves in hydraulically resistant porous media in a special parameter regime}}}}
\date{\today}

 \begin{abstract}{
 In this paper we study the stability of fronts  in a reduction of a well-known  PDE system that is used to model the combustion in hydraulically resistant porous media. More precisely, we consider the original  PDE system under the assumption that  one of the parameters of the model,  the Lewis number, is  chosen in a specific way and with initial conditions of a specific form.  {For a class of initial conditions, then the  number of unknown functions is  reduced from three to two. For the reduced system, the existence of combustion fronts follows from the existence results for the original system. The stability of these fronts is studied here by a combination of energy estimates and numerical Evans function computations and nonlinear analysis when applicable.  We then lift {\cre{the restriction}} on the initial conditions and show that the stability results obtained for the reduced system extend to the fronts in the full system considered for that specific value of the Lewis number.
 The fronts that we {investigate} are  proved to be either  absolutely unstable or   convectively unstable on the nonlinear level.}
}
\end{abstract}

\maketitle

\begin{center} 
Anna Ghazaryan \textsuperscript{a},  St\'ephane Lafortune \textsuperscript{b},  and Peter McLarnan\textsuperscript{a}\\
 \end{center}
 \textsuperscript{a} \address{ Department of Mathematics, Miami University, 301 S. Patterson Ave,   Oxford, OH 45056, USA,  Ph. 1-513-529-0582,}  
 \email{ghazarar@miamioh.edu}; \email{peterm@cs.earlham.edu} \par
 \textsuperscript{b} \address{ Department of Mathematics, College of Charleston, Charleston, SC 29424, USA, Ph.  1-843-953-5869},  \email{lafortunes@cofc.edu} \par



\section{Introduction}


Sivashinsky and his collaborators in \cite{BGSS}  proposed a model
describing combustion in inert porous media under condition of high hydraulic resistance, 
\eq{ 
T_t-(1-\gamma^{-1})P_t&=\epsilon T_{xx}+YF(T),  \\
P_t-T_t&=P_{xx}, \\
Y_t&=\epsilon {\rm{Le}}^{-1} Y_{xx}-\gamma YF(T), }{phys0}
where $Y$ is the   scaled  concentration of the reactant in the reaction zone, $P$ is the pressure, and 
$T$ is the temperature. The specific heat ratio
$\gamma>1$, the Lewis number ${\rm {Le}}>0$, and    the ratio of
pressure and molecular diffusivities  $\epsilon >0$ are physical characteristics of the fuel. The reaction rate  $YF(T)$ 
may or may not have an 
ignition cut-off, that is, $F(T) =0$   on an interval
$[0,T_{ign}] $ and  $F(T) >0$   and increasing for $T>T_{ign}$.
For   $F(T)$ Lipschitz continuity is assumed everywhere,  except {\cre{at}} the ignition temperature $T=T_{ign}$. 
Papers  \cite{BGSS,GR,BKS} contain detailed explanations and  the deduction of  this system.

In  \cite{Gordon_review} it is suggested to  considered  this system  with initial conditions
\begin{equation*}T(0,x)=T_0(x), \quad Y(0,x)=1, \quad P(0,x)=0. \label{PDEic}\end{equation*}
It is also assumed that  $\epsilon$  is significantly smaller than other parameters, therefore  a simplification of   
 \eqref{phys0} is offered in the literature which is obtained  by setting $\epsilon=0$, 
\eq{ 
T_t-(1-\gamma^{-1})P_t&=YF(T),  \\
P_t-T_t&=P_{xx}, \\
Y_t&=-\gamma YF(T). }{phys00}
We point out that  the terms $\epsilon T_{xx}$ and   $\epsilon {\rm Le}^{-1} Y_{xx}$ in (\ref{phys0})  are  singular perturbations, so  results obtained for the system \eqref{phys00} 
should  be dealt with caution, as singularly  perturbed systems in general may  support {\cre{a}} behavior significantly different than one observed in  the limiting system.

This paper is devoted   to the study of   the stability of traveling fronts in the system  \eqref{phys0}.
Traveling fronts  are solutions  of the underlying PDE \eqref{phys0}   that have a 
form $T(x,t)=T(\xi)$, $P(x,t)=P(\xi)$, $Y(x,t)=Y(\xi)$, with $\xi=x-ct$, where  $c$
is the a priori unknown front speed, and that asymptotically connect distinct  equilibria at $\pm \infty$.    These  solutions are sought as  solutions of the traveling wave ODEs 
\begin{eqnarray}
\label{systE}
-cT^{\prime}+c(1-\gamma^{-1})P^{\prime}&=&
\epsilon T^{\prime\prime}+YF(T), \nonumber \\
P^{\prime\prime}&=&c(T^{\prime}-P^{\prime}), \\
cY^{\prime}+\epsilon Y^{\prime\prime}&=&\gamma YF(T), \nonumber
\end{eqnarray}
that satisfy boundary-like conditions at $\pm \infty$ which we describe below.

Generally speaking, the equilibria of the system \eqref{systE} are states where $YF(T)=0$, so they can be described as  the states where there is no fuel $Y=0$, or as the cold states $T<T_{ign}$.  From physical considerations two of the equilibria are of interest, the completely burnt state
$P=1$, $ T=1$, $Y=0$,  and  the unburnt state where all of the fuel is present
$T=0$, $ P=0$, $Y=1$.  So we consider the system  \eqref{systE} together with the boundary conditions 
\eqnn{
P(-\infty)=1,\quad T(-\infty)=1, \quad Y(-\infty)=0, \\
T (+\infty)=0, \quad P(+\infty)=0,\quad Y(+\infty)=1. 
}

The existence  and uniqueness of  fronts   in  the system   (\ref{phys00}) has already  been established in  \cite{Dkh}, under  the assumption on the parameters  $0<T_{ign}<1-\gamma^{-1}.$
Moreover, it is already known  \cite{GKS},  that  the front  in  \eqref{phys00} that   satisfies the boundary conditions above is unique, up to translation.    
In  \cite{GR} it is shown that  as  $\epsilon $ approaches  $0$, the  $\epsilon$-dependent fronts  in (\ref{systE}) converge   to the fronts  of 
 \begin{eqnarray}
\label{syst0}
-cT^{\prime}+c(1-\gamma^{-1})P^{\prime}&=&
YF(T), \nonumber \\
P^{\prime\prime}&=&c(T^{\prime}-P^{\prime}), \\
cY^{\prime}&=&\gamma YF(T). \nonumber
\end{eqnarray}
 It is also known  \cite{GGJ}  that  the solution in the limiting system (\ref{syst0}) persists as a unique solution of speed of order $\rm O(1)$ of  the system (\ref{systE}) with $0<\epsilon\ll 1$.  We note that in \cite{GGJ} it is assumed that $\Le=1$, but this assumption can be removed because  it does not affect   the  result of the paper  in any way.

The stability of fronts in the system  \eqref{phys00}  has  been addressed  in \cite{GLM}.  There an Evans function approach was used to find parameter regimes   where the front  is absolutely unstable. In other words, it was shown that there are parameter regimes where small perturbations to the front   grow exponentially fast,  in the co-moving frame. More importantly, parameter regimes were found where  the front is convectively unstable, which means that  small perturbations to the front that are initially  localized near the rest state $(P, T, Y) =(1, 1,0)$ stay near that equilibrium. 

{{To our knowledge,  for no parameter values the stability of  the traveling fronts in \eqref{phys0} (with $\epsilon>0$) has  been yet  addressed. {\cre{We point out that}} since the perturbation with small $\epsilon$ is singular, the stability  (instability) of a front in the limiting system \eqref{phys00} does not directly imply the stability (instability)  of the front in \eqref{phys0}, even when $\epsilon$ is  very  small. {\cre{While we follow  the same standard  sequence of  steps  as we did in \cite{GLM} for the case $\epsilon=0$, from a technical point of view  our analysis is significantly different.}}
{\cre{Indeed, the case $\epsilon =0$ is a singular limit of Model (\ref{phys0}) in the sense that the order of the system is reduced by two (in \cite{GLM}, the order is furthermore reduced by one by using special initial conditions). It is well known that properties of existence and stability in a singular limit do not have to hold in general, even when $\epsilon$ is small (see for example \cite{Jones}). In this paper, we consider another singular limit for Model (\ref{phys0}). This limit is obtained by choosing  a particular value for the Lewis number, {in the presence of a {strictly} positive  $\epsilon$}, and the initial conditions to satisfy (\ref{incond2}) below. This choice reduces the order of Model (\ref{phys0}) by two. One main difference with the limit $\epsilon=0$ studied in \cite{GLM} is  that { in the reduction of  Model (\ref{phys0}) that we consider here}  the dimension of the linearization is larger, making the Evans function computation more complicated. Namely, we had to use the definition of the Evans function { that involves} the wedge product as opposed to the definition
corresponding to a linear system with a one-dimensional stable manifold as in \cite{GLM}, which is simpler and less numerically sensitive because the Evans function is defined by the scalar product of two solutions. 
Moreover, the nature of the energy estimates computed  in both papers is completely different. This is because the orders of the systems studied in \cite{GLM} and here {are} not the same.
}}

 As one would expect, the numerical calculation of the Evans function assumes a more precise definition of the reaction term than  the one given at the beginning of the introduction, therefore we base our analysis on  the assumption 
  \cite{W}  {{of}}
   a discontinuous  reaction rate 
   is 
\begin{equation*}
\label{Fdef}
F_{d}(T) = \left\{\begin{array}{ll}
\mathrm{exp}\left(Z\left\{\frac{T-h}{\sigma+(1-\sigma)T} \right\} \right), & T \geq T_{ign},\\
0, & T < T_{ign},  
\end{array} \right.
\end{equation*}
where $Z>0$ is the  Zeldovich number,  and $0<\sigma<1$ is the ratio of the characteristic temperatures of fresh and burned reactant.

The discontinuity in the reaction term is often introduced in combustion models \cite{Kh} to account for the fact that for low temperatures the reaction rate is many orders less than the reaction rate  at high temperatures.

However, to work with a well-defined linear operator obtained by linearizing the reaction term about the continuous front, we follow the recipe given in \cite{GLM}  and consider a smooth $F$, which is defined like $F_{d}$ everywhere except for a small interval $(T_{ign} , T_{ign}+2\delta)$ where the function is modified so as to go to zero in a smooth and monotonic fashion, for example, as  
\begin{equation}\label{Fdefrc}
F_\delta (v) = \left\{\begin{array}{ll}
\mathrm{exp}\left(Z\left\{\frac{v-h}{\sigma+(1-\sigma)v} \right\} \right), & v \geq T_{ign}+2\delta,\\
\\
\mathrm{exp}\left(Z\left\{\frac{v-h}{\sigma+(1-\sigma)v} \right\} \right)\,H^\delta(v-T_{ign}-\delta),&  T_{ign}\leq v < T_{ign}+2\delta,\\
\\
0, & v < T_{ign},
\end{array} \right.
\end{equation}
where 
$$H^\delta(x)=
\begin{cases} 
\frac{1}{1+{e^{{\frac {4x\delta}{\delta^2-x^2
 }}}}},\;\;{\mbox{for}}\;\;|x|<\delta,\\ 1,\;\;\;\; \,\;\;\;\;\; \quad {\mbox{for}}\;\; x\geq \delta,
 \\ 0\;\;\;\;\;\; \;\;\;\;\;\, \quad {\mbox{for}}\;\; x\leq -\delta,\end{cases} $$
or some other smooth approximation  of the Heaviside function $H$. In other words, $H^\delta$ is a function such that in the distributional sense
$
\lim_{\delta\rightarrow 0^+}H^\delta =H,$  and $H^\delta(x)=1$, for $x>\delta$, $H^\delta(x)=0$, for $x<-\delta$.

For numerical computations  in sections~ \ref{s:front} and \ref{num}, we choose $\delta$  small enough so that  the front velocity in the system with $F_{\delta}$ is  close to the velocity in the  system with  $F_d$. 
 
It is known \cite{GLM} that for the system with $\epsilon=0$  the front solution with $F=F_\delta$ converges  as $\delta\rightarrow 0^+$ to the front of the system with the reaction rate given by $F_d$.
The situation is more complicated when $\epsilon>0$.  The fronts in the  $\epsilon>0$ case exist  for any $\delta\geq 0$.   The proof is by construction and is based on geometric singular perturbation theory  which guarantees continuity in $\delta$ as long as $F_{\delta}$ is smooth, i.e.~for $\delta>0$.  
  The existing analytic proof does not provide information about whether the family of fronts parametrized by $\delta>0$ in the limit at $\delta \to 0^+$ converge to  the  wave that exists in the system with  discontinuous $F_d$.
We check numerically that it is indeed the case by verifying that for small values of $\delta$, we obtain wavespeeds very close to the wavespeeds of the discontinuous case. 

The stability results  in their turn depend on the way the  nonlinearity is smoothened and, in particular,  the value of $\delta$, but  we believe that the results that we obtain adequately reflect   what the stability of fronts in the  system \eqref{phys0} with a discontinuous reaction term will be.  
 
This paper  contains a   study  of the stability of the fronts in a reduced version of the system \eqref{phys0}  
which  is described below. The reaction term is taken with $F=F_\delta$ given in   \eqref{Fdefrc}.

\section{ Reduced model.}

There are two different ways to reduce  \eqref{phys0} that are discussed in literature. 
As we mentioned before, the standard reduction is  based on setting $\eps$ and $\eps \Le^{-1}$ equal to $0$. 

The other, different  reduction, that is described below, is based on  choosing a   value of $\Le$ in a very specific way.
The  reduced version of \eqref{phys0}  that is of interest to us here is suggested in \cite{GR, Gordon_review}. It is obtained  as follows.
 The system \eqref{phys0} is equivalent to the system 
\begin{eqnarray} \label{rr}
&& u_t= u_{zz}+yF(hu+(1-h) v), \nonumber \\
&& v_t= \eps v_{zz} +yF(hu+(1-h) v),  \\
&& y_t=\eps ({{\gamma}}(1-\mu){\Le})^{-1} y_{zz} -yF(hu+(1-h) v),\nonumber
\end{eqnarray}
 which is obtained by introducing  new variables  $(u,v,y)$  via a linear transformation
\eqnn{
T&= hu+  (1-h)v,\notag \\ P&= (1-\eps)^{-1}u -\eps (1-\eps )^{-1} v,\label{lintransf}\\Y& =y,\notag
}
where $h$ and $ \eps $ are defined as $\eps  = \epsilon\gamma (1-\mu)^2$, and $h=\mu/(1-\eps )=\mu/ (1- \epsilon(1-\mu)^2)$,   and  
$$\mu=\frac{\sqrt{\gamma^2(\epsilon+1)^2-4\gamma\epsilon}+\gamma(\epsilon-1)}{2\gamma \epsilon}, $$   
 and then rescaling of the independent variables
$$\tau=\gamma\,t,\;\;z=\sqrt{\gamma(1-\mu)}\,x.$$

The new parameters  in the system  \eqref{rr}  are assumed in \cite{Gordon_review} to have the following ranges:  $\eps \in(0,1)$ and $ h\in(0,1)$.   {We note that for $\epsilon \ll 1$, $\mu \approx  (\gamma -1)(\gamma -\epsilon)^{-1} $, and therefore $\eps\approx \epsilon\gamma$,
 $h \approx (1-\epsilon)^{-1} (\gamma -1) (\gamma -\epsilon)^{-1}\approx 1 - \gamma^{-1} +\rm{O}(\epsilon).$  In this paper we think of $\eps$ as a small parameter  and $h \in (0,1)$.}
 
{When ${{\Le^{-1}=\gamma(1-\mu)}}$, \eqref{rr} reads
\begin{eqnarray} \label{rrn}
&& u_t= u_{zz}+yF(hu+(1-h) v), \nonumber \\
&& v_t= \eps v_{zz} +yF(hu+(1-h) v),  \\
&& y_t=\eps  y_{zz} -yF(hu+(1-h) v).\nonumber
\end{eqnarray}
}
{We think of ${{\Le^{-1}=\gamma(1-\mu)}}$ as singular, because in this case,  if initially  
\eq{y(0,x) = 1-v(0,x),}{incond2} 
  then  
  $y(t,x) = 1-v(t,x)$ for  any $t >0$,  and $x\in \mathbb R$. 
  Therefore the system \eqref{rr}  reduces to 
\eq{
u_t&= u_{xx}+yF(hu+(1-h)(1-y)),  \\
y_t&= \eps y_{xx} -yF(hu+(1-h) (1-y)). 
}{r}
So the  system \eqref{r} is obtained by  choosing a specific  $\Le^*=\Le(
\mu,\gamma)$ in \eqref{phys0}  and  by considering  a specific set of initial conditions \eqref{incond2}.  
}

{ In more general setting which does not require initial conditions of any specific form, using a new variable {\cre{$g=v+y$}}, the system  \eqref{rrn} may be rewritten  as 
\begin{eqnarray} \label{rrnn}
&& u_t= u_{zz}+yF(hu+(1-h) v), \nonumber \\
&& {\cre{y_t= \eps y_{zz} -yF(hu+(1-h) v),}}  \\
&& g_t=\eps  g_{zz}.\nonumber
\end{eqnarray}
This paper is devoted to the  study of  the  stability of fronts  in the reduced system \eqref{rrnn}. Our strategy is  to establish stability results for  the fronts in \eqref{r} first and  then show in {\cre{Section \ref{DC}}} that these results extend to the system \eqref{rrnn}.} {\cre{In other words, we first use the initial condition \eqref{incond2} to simplify our computations but we are able to lift this restriction (without any additional computations) by using the fact that System \eqref{rrnn} differs from \eqref{r} only by the addition of the uncoupled equation $g_t=\eps  g_{zz}$.}}

 We point out that  considering a specific value of the Lewis number, { which is here slightly above $1$ when $\epsilon$ is small,}  in some sense reminds us about a situation with a well-known combustion model \cite{B, Varas02,EvansWeberGroup,BGHL}
\eq{ u_{t}&=u_{xx}+yF(u),\\ y_{t}&=\frac{1}{\Le} y_{xx}- Zy F(u),}{gasslessgeneral}
where $Z$ is exothermicity parameter.  When $\Le=1$,  the system can be reduced in a straightforward way to one leading PDE, as $Zu+y$ then simply satisfy a heat equation.  It is well known that, for some specific  functions $F(u)$, $\Le=1$ is a bifurcation value of the Lewis number when a transition from a unique traveling wave ($\Le >1$) solution to a multiple traveling solutions ($\Le<1$)  occurs \cite{BLS}.  To the knowledge of the authors there is no study of  a similar role of the critical value of  the Lewis number in \eqref{phys0}.  

\section{Front Solution \label{s:front}}

After introducing the  moving coordinate $\xi=x-ct$ in which the front is stationary, the system (\ref{r}) becomes
\begin{equation}
\begin{aligned}
u_t=&u_{\xi\xi} + cu_\xi + yF(hu+(1-h)(1-y)), \\
y_t=&\eps y_{\xi\xi} + cy_\xi - yF(hu+(1-h)(1-y)).\label{rm}
\end{aligned}
\end{equation}

As we have  discussed in the introduction there exists a value of $c$  such that  \eqref{rm} has
 a front solution, in other words  there is a $(\widehat{u}(\xi),\widehat{y}(\xi))$  that solves
 \begin{equation}
\begin{aligned}
u_{\xi\xi} + cu_\xi +yF(hu+(1-h)(1-y))=0,\\
\eps y_{\xi\xi} + cy_\xi -yF(hu+(1-h)(1-y))=0,\label{tw}
\end{aligned}
\end{equation}
 and   satisfies
 \eq{(u,y) \to (1,0) \text{ as } \xi\to -\infty,\\  (u,y) \to (0,1) \text{ as }  \xi\to+\infty.}{bc0}

Numerically, a solution to \eqref{tw} satisfying boundary conditions \eqref{bc0} can be found as follows. We reduce the order of  the system \eqref{tw} by introducing the following variable
$$
z = \frac{1}{c}\int_\xi^\infty yF(\ww).
$$
System \eqref{tw} then becomes
\begin{equation*}
\begin{aligned}&u_{\xi\xi}+cu_\xi-cz_\xi = 0,\\
&\eps y_{\xi\xi}+cy_\xi-cz_\xi = 0,
\end{aligned}
\end{equation*}
which, after taking into account the boundary conditions \eqref{bc0},  can be integrated to obtain
\begin{equation*}
\begin{aligned}&u_\xi + cu - cz = 0, \\
&\eps y_\xi + c y - cz = c.
\end{aligned}
\end{equation*}
We then are  left with the system
\begin{equation}
\begin{aligned} u_\xi &= c(z-u),  \\
y_\xi &= \frac{c}{\eps}(1- y + z), \\
z_\xi &= -\frac{1}{c}y F(\ww),
\end{aligned} \label{e:z}
\end{equation}
with boundary conditions  $(u,y,z) \to (1,0,1)$ as $\xi\to -\infty$, and $(u, y,z) \to (0,1,0)$ at $\xi\to \infty$.
We are thus interested in the heteroclinic orbit that connects the the fixed points $(1,0,1)$ and $(0,1,0)$. There is one unstable direction  at the point  $(1,0,1)$.  
  The corresponding eigenvalue and eigenvector are given by
\eqnn{&\mu=\frac{-c+\sqrt{c^2+4\eps e^{(1-h)Z}}}{2\eps},\\
\;\;\;&{\cre{v_u=\left(1,\left( c(\eps-1)\mu e^{-(1-h)Z}+1\right)/\eps,\;
\mu/c+2  \right).}}}
To find the front solution we use a simple shooting method, that is we use initial conditions such that {\cre{$\left(u(0)-1,v(0),z(0)-1\right)=\beta v_u$ for a small $\beta>0$. For given values of  $\eps$, $h$, $\sigma$, $\delta$, $T_{ign}$, and $Z$, we use such initial conditions  and integrate System \eqref{e:z} using the numerical integrator {\sl{ODE45}} from MatLab for various values of 
 $c$ (with the {\sl{ODE45}} absolute and relative error tolerances set to $10^{-15}$ and $10^{-5}$, respectively). We do this until  an appropriate connection to the fixed point (0,1,0) is found. The $\beta$ is chosen small enough so that the speed $c$ appears to have converged, that is $c$ does not change substantially when $\beta$ is made smaller. We found that choosing $\beta$ to be $10^{-8}$ works well. Indeed, in all the cases treated in this article, we have checked that the percent change in the value of $c$ does not exceed 0.2\% when $\beta$ is changed from $10^{-8}$ to $10^{-9}$.}}
 
 \cre{{As an example of computation}},  Figure \ref{FigFront} shows the front solution as a function of $\xi$ in the case where $\eps=0.1$, $h=0.3$, $\sigma=0.25$, $\delta=0.0005${\cre{,}} $T_{ign}=0.01$, and $Z=6$, which corresponds to the speed $c=1.8588$. 
  
{{Figure \ref{SpeedFigure} shows how the speed $c$ varies as a function of the various parameters. In particular, Figure \ref{SpeedFigure} (d) shows $c$ as a function of  $\delta$. The parameter $\delta$ is chosen so that the speed $c$ is close to the speed of the discontinuous system. For this purpose, we use $\delta=0.0005$ as we have checked that with $\delta=0.0005$, the percentage difference between the velocities of the smoothed and the discontinuous  models do not exceed 2.2\% in all the cases considered in this article.

\section{Linearization \label{s:lin}}

To analyze the stability of the front we consider  the linearization
of (\ref{rm}) about the front $(\widehat{u},\widehat{y})$. 
 The eigenvalue problem for the linearized operator 
reads
\begin{equation}
\label{EigProb}
\begin{aligned}\lambda p &= p_{\xi\xi} + c p_\xi + F_w(\hw)\,\hy\,(hp-(1-h)q)+F(\hw)q, \\
\lambda q &= \eps q_{\xi\xi} + c q_\xi-F_w(\hw)\,\hy\, (hp-(1-h)q)-F(\hw)q,
\end{aligned}
\end{equation}
where    $w=hu+(1-h)(1-y)$ and $\hw = w(\widehat{u},\widehat{y})$, and  subscript $w$ denotes the derivative with respect to $w$.
In a reaction-diffusion equation, a wave  is spectrally stable  in a space if the spectrum of (\ref{EigProb})
of the linearization of the system about the wave in that space  is contained in the half-plane
$\{\realpart{ \lambda}  \leq -\nu\}$  for some $\nu > 0$ with the exception of a simple zero eigenvalue which is generated by the translational invariance.

\subsection{Essential Spectrum\label{S:ess}}

Let ${\mathcal{L}}$ be the operator defined by the right hand side of \eqref{EigProb}.
To study the essential spectrum of ${\mathcal{L}}$ on {\cre{the}} space $L^2$, we  consider  operators  obtained by taking the limits of ${\mathcal{L}}$ at $\xi \to\pm \infty$, in other words, the operators obtained  from ${\mathcal{L}}$ by  inserting  {\cre{$\hy=1$, $\hw=0$,  
$\hy=0$, and $\hw=1$}},
$${\mathcal{L}}_{\infty}=\begin{pmatrix}\frac{\rd^2}{\rd \xi^2}+c\frac{\rd}{\rd \xi}& 0 \\0 & \eps\frac{\rd^2}{\rd \xi^2}+c\frac{\rd}{\rd \xi}\end{pmatrix},\;\;{\mathcal{L}}_{-\infty}= \begin{pmatrix}\frac{\rd^2}{\rd \xi^2}+c \frac{\rd}{\rd \xi}& e^{(1-h)Z} \\0 &\eps\frac{\rd^2}{\rd \xi^2}+ c\frac{\rd}{\rd \xi}-e^{(1-h)Z}\end{pmatrix}.$$
The spectra of  the  operators ${\mathcal{L}}_{\pm\infty}$  \cite{Henry81}  is a set of curves that  are included in  the essential spectrum of ${\mathcal{L}}$. Moreover, the essential spectrum of  ${\mathcal{L}}$ is bounded by these curves.
One can find  these curves  from the  linear dispersion relation obtained by using modes of the form $e^{\i \sigma\xi}$. {\cre{ This amounts to considering the quantities ${\mathcal{L}}_{\pm\infty}(\sigma)$ obtained from ${\mathcal{L}}_{\pm\infty}$ by replacing the derivative with respect to $\xi$ by $\sigma$, which in its turn implies that the spectra can be obtained by working with  the equations 
$\label{det}{\mbox{det}}\left({\mathcal{L}}_{\pm\infty}(\sigma) -\lambda I\right)=0$ (see \cite{Simon06} for details). 
} 
It is found  in \cite{Simon06}  that since $\eps <1$, the rightmost curve defined by (\ref{det}) is $\lambda=-\eps \sigma^2+\ri  c \sigma$, or $$\realpart \lambda=-\frac{\epsilon}{c^2}\, \impart \lambda ^2.$$
 The region to the right of this parabola and therefore the open left side of the complex plane (with the exception of the origin) can only contain point spectrum, i.e.~eigenvalues with finite multiplicity. 

On the other hand,  the essential spectrum on the subspace of $L^2$ that consists of functions that 
decay on $+\infty$ faster than $e^{\alpha \xi}$ for some $\alpha >0$ lays strictly to the left of the imaginary axes. Indeed, substituting  $(P,Q)=(p,q)e^{\alpha\xi}$ in \eqref{EigProb} we see that the linear operators produced by the linearization about the equilibria are given by
\eqnn{
{\mathcal{L}}_{\infty}(\alpha)&=\begin{pmatrix}\frac{\rd^2}{\rd \xi^2}+(c-2\alpha)\frac{\rd}{\rd \xi} +(\alpha^2-c\alpha)& 0 \\0 & \eps\frac{\rd^2}{\rd \xi^2}+(c-2\alpha)\frac{\rd}{\rd \xi}+(\epsilon\alpha^2-c\alpha)\end{pmatrix},\\
{\mathcal{L}}_{-\infty}(\alpha)&= \begin{pmatrix}\frac{\rd^2}{\rd \xi^2}+(c-2\alpha) \frac{\rd}{\rd \xi} +(\alpha^2-c\alpha) & e^{(1-h)Z} \\0 &\eps\frac{\rd^2}{\rd \xi^2}+ (c-2\alpha)\frac{\rd}{\rd \xi}-e^{(1-h)Z}\ +(\eps\alpha^2-c\alpha)\end{pmatrix}.
}
Therefore on that weighted space, the right most boundary of the essential spectrum is given by 
$\lambda=-\eps \sigma^2+\ri  (c-2\alpha) \sigma +(\alpha^2-c\alpha)$, or 
$\realpart  \lambda=-\frac{\eps}{(c-2\alpha)^2}\, \impart \lambda ^2 +(\alpha^2-c\alpha).$
So as long as  $0<\alpha<c/2$, { the sign of the convection term $c-2\alpha$ is preserved and  $\alpha^2-c\alpha<0$, the essential spectrum is to the left of the imaginary {\cre{axis}} and its boundary has the same orientation as the boundary of the essential spectrum of the  original problem.}

\subsection{The point spectrum.}

Further in the paper, we are planning to use numerical Evans function calculation to look for unstable eigenvalues of  (\ref{EigProb}). {\cre{Before doing so, we first want to}}  identify a bounded region where those eigenvalues may  reside, or, in other words, to identify a subset of the right half-plane of the complex plane with a bounded complement  where the unstable discrete eigenvalues are guaranteed to be absent. We accomplish this by obtaining  a bound on eigenvalue  $\lambda$ of  the operator ${\mathcal{L}}$  under the assumption $\realpart{\lambda}\geq 0$. 
The  following lemma  holds for such $\lambda$.
\begin{lemma} \label{L:4.1d}
\eq{\realpart{\lambda}\iR{|p|^2}\leq h\iR{ F_w(\hw)\,\hy\,|p|^2}+\iR{(F(\hw)+(1-h)\, F_w(\hw)\,\hy)\,\left(\frac{|q|^2}{4\epsilon_1}+\epsilon_1 |p|^2\right)},}{rein}
 and
\eq{\left(\realpart{\lambda}+|\impart{\lambda}|\right)\iR{|p|^2}&\leq  \frac{c^2}{4}\iR{|p|^2 }+h\iR{ F_w(\hw)\,\hy\,|p|^2}\\&+ \iR{\left({F(\hw)+(1-h)  F_w(\hw) \,\hy}\right)\left(\frac{|q|^2}{2\epsilon_1}+\epsilon_1|p|^2\right)},}{imRinu}
for any real positive value of $\epsilon_1$.
\end{lemma}

\begin{proof}
From the definition of $F$ it is easy to see that 
 $F$  is a nondecreasing function of $w$ for all fixed   $\sigma$ and $h$ in $[0,1]$, so  its derivative is nonnegative. 
We multiply the first equation of (\ref{EigProb}) by $\overline{p}$ and we then integrate to obtain
\begin{equation}
\label{fe}
\lambda \iR{|p|^2}=c\iR{p_\xi \overline{p}} -\iR{|p_\xi|^2}-(1-h)\iR{F_w(\hw)\,\hy \,q\,\op}+\iR{F(\hw)\,q\,\op}+h\iR{F_w(\hw)\,\hy\,|p|^2}.
\end{equation}
We take the real and imaginary parts of (\ref{fe}) to get
\eq{\realpart{\lambda}\iR{|p|^2}=h\iR{F_w(\hw)\,\hy\,|p|^2}-\iR{|p_\xi|^2}+\iR{\left(F(\hw)-(1-h){F_w(\hw)\,\hy }\right)\,\realpart{q\,\op}},}{re}
and
\eq{\impart{\lambda}\iR{|p|^2}=-\ri c\iR{p_\xi \overline{p}}+\iR{\left(F(\hw)-(1-h){F_w(\hw)\,\hy }\right)\,\impart{q\,\op}}.}{im}
Using Young's inequality on (\ref{re}), we  obtain (\ref{rein}).

From (\ref{im}), we get the following inequality
\eq{|\impart{\lambda}|\iR{|p|^2}\leq  c\iR{|p_\xi| |{p}|}+\iR{\left(F(\hw)+(1-h){F_w(\hw)\,\hy }\right)\,|\impart{q\,\op}|}.}{imin}
Adding (\ref{re}) and (\ref{imin}), we obtain
\eq{\left(\realpart{\lambda}+|\impart{\lambda}|\right)\iR{|p|^2}\leq  c\iR{|p_\xi| |{p}|}-\iR{|p_\xi|^2}+h\iR{F_w(\hw)\,\hy\,|p|^2}\\+ \iR{\left(F(\hw)+(1-h){F_w(\hw)\,\hy }\right)\,\left(\realpart{q\,\op}+|\impart{q\,\op}|\right)}.}{imRin}
Applying Young's inequality  ${c\,|p_\xi ||p|\leq \frac{c^2|p|^2}{4}+|p_\xi|^2}$ to the first term   of the RHS of (\ref{imRin})  and using the fact that 
$$\realpart{q\,\op}+|\impart{q\,\op}|\leq \sqrt{2}|p||q|\leq \frac{|q|^2}{2\epsilon_1}+\epsilon_1|p|^2,$$ 
we find (\ref{imRinu}). 
\end{proof} 
We next prove the following lemma.
\begin{lemma} \label{L:4.2}
\eq{\realpart{\lambda}\iR{|q|^2} + \iR{\left(F(\hw)-(1-h)F_w(\hw)\,\hy \right)|q|^2}\leq h\iR{F_w(\hw)\,\hy\,\left(\frac{|q|^2}{4\epsilon_2}+\epsilon_2 |p|^2\right)},}{rein2}
 and
\eq{\left(\realpart{\lambda}+|\impart{\lambda}|\right)\iR{|q|^2}&\leq \frac{c^2}{4\eps}\iR{|q|^2 }+h\iR{F_w(\hw)\,\hy\,\left(\frac{|q|^2}{2\epsilon_2}+\epsilon_2 |p|^2\right)}\\&- \iR{\left({F(\hw)-(1-h)F_w(\hw)\,\hy}\right)|q|^2},}{imRinu2}
for any real positive value of $\epsilon_2$
\end{lemma}

\begin{proof}
We multiply the second equation of (\ref{EigProb}) by $\overline{q}${\cre{, we integrate and get}}
\begin{equation}
\label{fe2}
\lambda \iR{|q|^2}=c\iR{q_\xi \overline{q}}-\eps \iR{|q_\xi|^2}- \iR{\left(F(\hw)-(1-h)F_w(\hw)\,\hy \right)|q|^2}-h\iR{F_w(\hw)\,\hy\,p\,\oq}.
\end{equation}
The real and imaginary parts of (\ref{fe2}) are
\eq{\realpart{\lambda}\iR{|q|^2}=-\eps \iR{|q_\xi|^2}- \iR{\left(F(\hw)-(1-h)F_w(\hw)\,\hy \right)|q|^2}-h\iR{F_w(\hw)\,\hy\,\realpart{p\,\oq}}.}{re2}
and
\eq{\impart{\lambda}\iR{|q|^2}=-\ri c\iR{q_\xi \overline{q}}-h\iR{F_w(\hw)\,\hy\,\impart{p\,\oq}}.}{im2}
Using Young's inequality on (\ref{re2}), we get (\ref{rein2}).

From (\ref{im2}), we get the inequality
\eq{|\impart{\lambda}|\iR{|q|^2}\leq c\iR{|q_\xi| |{q}|}+h\iR{F_w(\hw)\,\hy\,\left| \impart{p\,\oq}\right|}.}{imin2}
Adding (\ref{re2}) and (\ref{imin2}), we obtain
\eq{\left(\realpart{\lambda}+|\impart{\lambda}|\right)\iR{|q|^2}\leq  c\iR{|q_\xi| |{q}|}+h\iR{F_w(\hw)\,\hy\,\left(\left|\impart{p\,\oq}\right|-\realpart{p\,\oq}\right)}
\\-\eps \iR{|q_\xi|^2}- \iR{\left(F(\hw)-(1-h)F_w(\hw)\,\hy \right)|q|^2}.}{imRin2}
Applying Young's inequality to $c|q||q_{\xi}|$ in 
 the first term of the RHS of (\ref{imRin2}) in the form
$$c|q||q_{\xi}| \leq \frac{c^2 |q|^2}{4 \eps} +\eps |q_{\xi}|^2,$$
and the fact that $$\left|\impart{q\,\op}\right|-\realpart{q\,\op}\leq \sqrt{2}|p||q|\leq \frac{|q|^2}{2\epsilon_2}+\epsilon_2 |p|^2,$$ we find (\ref{imRinu2}).
\end{proof}

The following theorem  holds.
\begin{theorem}
\label{T1d}
If $\lambda\neq 0$ is an eigenvalue for (\ref{EigProb}) with nonnegative real part, then we have the following inequalities:
\eq{\realpart{\lambda}\leq \left(3-\frac{3h}{2}\right)\sup_{\mathbb R}( F_{{w}}(\hw)\,\hy) +\mathrm{exp}\left(Z(1-h)\right),}{rein126}
and
\eq{|\lambda|\leq \frac{c^2}{4}\max{\left\{1,\,\frac{1}{\eps}\right\}}+\left(3-h\right)\sup_{\mathbb R}( F_{{w}}(\hw)\,\hy)+
\mathrm{exp}\left(Z(1-h)\right).}{imRinu7}
\end{theorem}

\begin{proof}
Multiplying (\ref{rein}) by $\theta$ and {\cre{adding}} (\ref{rein2}), we obtain
\eqnn{\realpart{\lambda}\iR{(\theta |p|^2+|q|^2)}\leq \left((\theta+\epsilon_2)\,h+\epsilon_1\theta(1-h)\right)\iR{F_w(\hw)\,\hy\,|p|^2}\\+\left(1-h+\frac{h}{4\epsilon_2}+\frac{\theta(1-h)}{4\epsilon_1}\right)\iR{F_w(\hw)\,\hy\,|q|^2}+\\\left(\frac{\theta}{4\epsilon_1}-1\right)\iR{F(\hw)\,|q|^2}+\epsilon_1\theta\iR{F(\hw)|p|^2}.}
Next we choose  $\theta=4$ and  $\epsilon_1=\epsilon_2=1$ and obtain
\eq{\realpart{\lambda}\iR{(4 |p|^2+|q|^2)}\leq \left(\frac{h}{4}+1\right)\iR{F_w(\hw)\,\hy\,4|p|^2}\\
+\left(2-\frac{7h}{4}\right)\iR{F_w(\hw)\,\hy\,|q|^2}+
\iR{F(\hw)\,4 |p|^2}.}{rein124}
Since   $F$ is nondecreasing,  \eq{\displaystyle{\sup_{x\in \mathbb{R}}\left(F(\hw)\right)=F(1)=\mathrm{exp}\left(Z(1-h)\right)}.}{Fsupd}
The inequality (\ref{rein126}) then follows from using (\ref{Fsupd})  in (\ref{rein124}).

We then multiply (\ref{imRinu}) by $\theta$ and add (\ref{imRinu2}) to  obtain
\eqnn{&\left(\realpart{\lambda}+|\impart{\lambda}|\right)\iR{\left(\theta|p|^2+|q|^2\right)}\\ &\leq  \frac{c^2}{4}
\max{\left\{1,\,\frac{1}{\eps}\right\}}
\iR{\left(\theta|p|^2+|q|^2\right) }+ \left(\frac{\theta}{2\epsilon_1}-1\right)\iR{{F(\hw)}|q|^2}\\ &+\epsilon_1\iR{{F(\hw)}\,\theta\, |p|^2}+\left(\frac{h\epsilon_2}{\theta}+h+(1-h)\epsilon_1\right)\iR{F_w(\hw)\,\hy\,\theta\,|p|^2}\\ &+\left(\frac{\theta(1-h)}{2\epsilon_1}+\frac{h}{2\epsilon_2}+1-h\right)\iR{F_w(\hw)\,\hy\,|q|^2}.}
Here, we choose $\theta=2$ and $\epsilon_1=\epsilon_2=1$ and get
\eqnn{&\left(\realpart{\lambda}+|\impart{\lambda}|\right)\iR{\left(2|p|^2+|q|^2\right)}
\leq  \frac{c^2}{4}
\max{\left\{1,\,\frac{1}{\eps}\right\}}
\iR{\left(2|p|^2+|q|^2\right) }\\
&+\iR{{F(\hw)}\,2\, |p|^2}+\left(\frac{h}{2}+1\right)\iR{F_w(\hw)\,\hy\,2\,|p|^2}+\left(2-\frac{3h}{2}\right)\iR{F_w(\hw)\,\hy\,|q|^2}.}
Finally, using (\ref{Fsupd}), we obtain
\eqnn{\left(\realpart{\lambda}+|\impart{\lambda}|\right)\leq \frac{c^2}{4}
\max{\left\{1,\,\frac{1}{\eps}\right\}}
+\left(3-h\right)M+\mathrm{exp}\left(Z(1-h)\right).}
The fact that $|\lambda |\leq \realpart{\lambda}+|\impart{\lambda}|$ when $\realpart{\lambda}\geq 0$ then implies (\ref{imRinu7}). \end{proof}

\section{Linear Stability \label{num}}

The system (\ref{EigProb}) can be turned into a linear dynamical system 
of the form
\begin{equation}
X'=A(\xi,\lambda)\,X,
\label{linear}
\end{equation}
where $A$ is the $4\times 4$ following square matrix 
\eq{
A(\xi,\lambda)=
\left(
\begin{array}{cccc}
0 & 1 & 0 & 0 \\
\lambda - h\,F_w(\hw)\,\hy  &  -c  &  (1-h)\, F_w(\hw)\,\hy-F(\hw)  & 0 \\
0 & 0 & 0 & 1 \\
h\, F_w(\hw)\,\hy/\eps &  0  &  \left(\lambda + (1-h)\,F_w(\hw)\hy+F(\hw)\right)/\eps  &  - c / \eps
\end{array}
\right).
}{A}
The asymptotic behavior as $\xi\rightarrow\infty$ of the solutions to (\ref{linear}) 
is determined by the matrices
\[
{\mathcal A}^{\infty}(\lambda)=\lim_{\xi\rightarrow \infty}A(\xi,\lambda),
\]
which is found by inserting the values $\hy=1$ and $\hw=0$ into (\ref{A})
\eq{
{\mathcal A}^{\infty}(\lambda) = \begin{pmatrix}
0 & 1 & 0 & 0 \\
\lambda &  -c  &  0  & 0 \\
0 & 0 & 0 & 1 \\
0 &  0  &  \lambda /\eps  &  - c / \eps
\end{pmatrix}.
}{inf}
For $\realpart{\lambda} > 0$, the matrix of ${\mathcal A}^{\infty}$ has two eigenvalues with negative real part. They are given by
\begin{equation} \mu_{1+} = -\frac{1}{2\eps} (c + \sqrt{c^2 + 4 \eps \lambda}),\,\,\,
    \mu_{2+} = -\frac{1}{2} (c + \sqrt{c^2 + 4 \lambda}), \label{eiinf} \end{equation}
    and their corresponding eigenvectors are
$$ v_{1+} =  \left(0,0,1,\mu_{1+}\right)^T,\;\;\;\;
    v_{2+} =  \left(1,\mu_{2+},0,0\right)^T.$$
The asymptotic behavior as $\xi\rightarrow -\infty$ of the solutions to (\ref{linear}) 
is determined by the matrix
\eqnn{
{\mathcal A}^{-\infty}(\lambda)=\lim_{\xi\rightarrow -\infty}A(\xi,\lambda),
}
which is found by inserting the values $\hy=0$ and $\hw=1$ into (\ref{A})
\eq{
{\mathcal A}^{-\infty}(\lambda) = \begin{pmatrix}
0 & 1 & 0 & 0 \\
\lambda &  -c  &  e^{(1-h)Z}  & 0 \\
0 & 0 & 0 & 1 \\
0 &  0  &  (\lambda + e^{(1-h)Z})/\eps  &  - c / \eps
\end{pmatrix}.
}{minf}
For $\realpart{\lambda} > 0$, the matrix ${\mathcal A}^{-\infty}$ has two eigenvalues with positive real part. They are given by
\eq{  \mu_{1-} =-\frac{1}{2\eps} \left(c  - \sqrt{c^2 + 4 \eps \lambda  +
4 \eps e^{Z(1-h)}}\right) ,\quad
    \mu_{2-} = -\frac{1}{2}\left(c  - \sqrt{c^2 + 4 \lambda}\right),}{eiminf}
    and their corresponding eigenvectors are
\eqnn{
 v_{1-} = 
 & \left(1, \mu_{1-}, \frac{(1-\eps)(\lambda-c\,\mu_{1-})}{\eps\,e^{Z(1-h)}}+\frac{1}{\eps},
 {\frac { \left(  \left( {c}^{2}+\lambda\,\eps \right)  \left( 1-
\eps \right) +\eps\,e^{Z(1-h)} \right) {\it \mu_{1-}}-c \left( \lambda+e^{Z(1-h)}
 \right)  \left( 1-\eps  \right) }{{\eps}^{2}e^{Z(1-h)}}}
\right)^T,\\
v_{2-}=&\left(1,\,\mu_{2-},0,0\right)^T.
}

As a consequence \cite{codd}, the system (\ref{linear}) has two linearly independent solutions $X_{1+}$ and $X_{2+}$ converging to zero as $\xi\rightarrow\infty$ 
and two solutions $X_{1-}$ and $X_{2-}$ converging to zero as $\xi\rightarrow -\infty$, satisfying
$$\lim_{\xi\rightarrow \pm \infty} X_{i\pm} e^{-\mu_{i\pm}\xi}=v_{i\pm},\;\;i=1,2.$$
Clearly, a value of $\lambda$ is an eigenvalue for the problem (\ref{EigProb}) if and only if  the
 space of solutions of (\ref{linear}) bounded as $\xi \rightarrow + \infty$, spanned by $\{X_{1+},\,X_{2+}\}$, and
the space of solutions bounded as $\xi \rightarrow - \infty$, spanned
by $\{X_{1-},\,X_{2-}\}$, have an intersection of strictly positive dimension.
Those values of $\lambda$ can be located with the help of the Evans function \cite{Evans,Jones,Yanagida,Alexander90,Pego,Gardner98,Kapitula98a,Sandstede,Kapitula00,Li00}. 
The Evans function is a function of the spectral parameter $\lambda$; it is analytic,  real for $\lambda$ real, and it vanishes on the point
spectrum of the problem (\ref{EigProb}). 
There are several definitions of the Evans function. For our purpose of numerical computations, we use the
definition involving exterior algebra
\cite{AfBr01,AlBr02,Br99,BrDeGo02,Br00,DeGo05,NgR79,skms,EvansWeberGroup}.

We are  in the situation where the dimension of the system is $n = 4$ and the dimensions of the stable and unstable manifolds are  $n_s = n_u = 2$. In such a case, we consider the
wedge-space~$\CwedgeX$, the space of all two forms on ${\mathbb C}^{4}$. 
The induced dynamics of (\ref{linear}) on~$\CwedgeX$ can be written as
\begin{equation}\label{Unk}
{U}' = {\bf A}^{(2)}(\xi,\lambda){U} .
\end{equation}
Here the matrix ${A}^{(2)}$ is matrix generated by  ${A}=\{a_{ij}\}$ on the wedge-space $\CwedgeX$. Using the standard
basis of $\CwedgeX$,
\begin{equation}\label{5.1}
\begin{array}{l}
{\bf \omega}_1 = {\bf e}_1\wedge{\bf e}_2 ,\quad
{\bf \omega}_2 = {\bf e}_1\wedge{\bf e}_3 ,\quad
{\bf \omega}_3 = {\bf e}_1\wedge{\bf e}_4 ,\\
{\bf \omega}_4 = {\bf e}_2\wedge{\bf e}_3 ,\quad
{\bf \omega}_5 = {\bf e}_2\wedge{\bf e}_4 ,\quad
{\bf \omega}_6 = {\bf e}_3\wedge{\bf e}_4 ,
\end{array}
\end{equation}
where $\left\{{\bf e}_{1},{\bf e}_{2},{\bf e}_{3},{\bf e}_{4}\right\}$ is the standard basis of ${\mathbb C}^{4}$,
the matrix ${A}^{(2)}$ is given
by
{\small\[
{A}^{(2)}=
\left(
\begin{array}{ccccccccccc}
a_{11}\!+\! a_{22} && a_{23} && a_{24} && -a_{13} && -a_{14} && 0\\
a_{32} && a_{11}\!+\! a_{33} && a_{34} && a_{12} && 0 && -a_{14}\\
a_{42} && a_{43} && a_{11}\!+\! a_{44} && 0 && a_{12} && a_{13}\\
-a_{31}&& a_{21}&& 0 && a_{22} \!+\! a_{33} && a_{34} && -a_{24}\\
-a_{41}&& 0 && a_{21} && a_{43} && a_{22} \!+\! a_{44} && a_{23}\\
0  && -a_{41}&& a_{31}&& -a_{42} && a_{32} && a_{33} \!+\! a_{44}
\end{array}
\right).
\]}
In our case, the asymptotic matrices are given by
\[
\lim_{\xi\to\pm\infty}{A}^{(2)}(\xi,\lambda)=
\left({\mathcal A}^{\pm \infty}\right)^{(2)},
\]
where ${\mathcal A}^{\pm \infty}$ are as in (\ref{inf}) and (\ref{minf}). The eigenvalue of 
$\left({\mathcal A}^{\infty}\right)^{(2)}$ with the smallest real part is $\mu_{1+}+\mu_{2+}$ with eigenvector $v_{1+}\wedge v_{2+}$. 
The solution of (\ref{Unk}) given by $U_+=X_{1+}\wedge X_{2+}$ then behaves as
$$
\lim_{\xi\rightarrow \infty} U_{+} e^{-(\mu_{1+}+\mu_{1+})\xi}=w_+\equiv v_{1+}\wedge v_{2+}.
$$
Similarly, the solution  $U_-=X_{1-}\wedge X_{2-}$ behaves as
$$
\lim_{\xi\rightarrow -\infty} U_{-} e^{-(\mu_{1-}+\mu_{-})\xi}=w_-\equiv v_{1-}\wedge v_{2-}.
$$
This allows us to define the Evans function as
%
$$E(\lambda) \equiv {
U}_-\wedge {U}_+,$$
{\cre{where $U_{\pm}$ are evaluated at some chosen value of $\xi$ (often taken to be $\xi=0$)}}. If $U$ is written in components in the basis (\ref{5.1}),  then the Evans function is computed  \cite{AlBr02,BrDe99,BrDeGo02} as
$$E(\lambda) = {{
U}_-^T}\, \Sigma\, {U}_+,$$
where ${
\Sigma}$ is the matrix
\[
{\bf \Sigma} =
\left[
\begin{array}{ccccccccc}
0 & \hfill0 & \hfill0 & \hfill0 & \hfill0 & 1 \\
0 & \hfill0 & \hfill0 & \hfill0 & -1      & 0 \\
0 & \hfill0 & \hfill0 & \hphantom{-}1 & \hfill0 & 0 \\
0 & \hfill0 & \hphantom{-}1  & \hfill0 & \hfill0 & 0 \\
0 & -1 & \hfill0 & \hfill0 & \hfill0 & 0 \\
1 & \hfill0 & \hfill0 & \hfill0 & \hfill0 & 0
\end{array}
\right].
\]

 The function $E(\lambda)$ will be analytic in the any region of the complex plane where the eigenvalues 
$\mu_{1+}+\mu_{2+}$ and $\mu_{1-}+\mu_{2-}$ are, respectively, the eigenvalues with smallest and largest real part of $\left({\mathcal A}^{\infty}\right)^{(2)}$ and $\left({\mathcal A}^{- \infty}\right)^{(2)}$. In view of the expressions for the eigenvalues given in  (\ref{eiinf}) and (\ref{eiminf}), to define such a region, it suffices to implement the condition
\begin{equation}
\realpart{\lambda}>-\frac{1}{4}\min{\left(1,\,\frac{1}{\eps}\right)}.
\label{ancondition}
\end{equation}


{\cre{To compute the Evans numerically, we choose a positive value $\xi=L$ at which the matrix given $A$ is suitably close to its asymptotic value ${\mathcal{A}}^{\infty}$. We then use {\sl{ODE45}} to integrate the system (\ref{Unk}) backward from $\xi=L$ in the direction of the eigenvector $w_+$ and find $U_+$ (with the {\sl{ODE45}} absolute and relative error tolerances both set to $10^{-10}$). Additionally, in order to eliminate the exponential growth due to the eigenvalue with negative real part as we integrate from $\xi=L$, we modify the system (\ref{Unk}) in the following way  
\begin{equation}
\label{Unkm}
 U'=\left(A^{(2)}-(\mu_{1+}+\mu_{2+}) I\right)\,U,
 \end{equation}
 and use the initial condition $U(L)=w_+$. 
To find $U_-$, we integrate the following modified version of  (\ref{Unk})
\begin{equation}
 U'=\left(A^{(2)}-(\mu_{1-}+\mu_{2-}) I\right)\,U,
 \end{equation}
frontward from $\xi=0$ with initial condition $U(0)=w_-$, (recall that $\xi=0$ is the value of $\xi$ from which the front was computed (see Section \ref{s:front})).  The Evans function then is taken to be the quantity   
${{U}_-^T}\, \Sigma\, {U}_+$ evaluated at the value $\xi=\xi^*$ at which the front solution of (\ref{e:z}) satisfies  $u(\xi^*)=y(\xi^*)$ (in Figure \ref{FigFront}, this corresponds to the value of $\xi$ at which the solid and dashed lines meet). Note that in the numerical computations, we choose the eigenvectors $w_\pm$ so that $ {{
w}_-^T}\, \Sigma\, {w}_+=1$. 
{Note also that in our case, since $F(w)=0$ for {$w< v_{ign}$}, the matrix $A$ in \eqref{A} reaches the constant matrix $A^{\infty}$ for a finite positive value of $\xi$. When computing $X_+$, we thus choose the value of $L$ to be greater than the lowest value of $\xi$ satisfying $\hw(\xi)<T_{ign}$.}}}

Since we are interested in the zeroes of the Evans function, the standard method is to compute the integral of the logarithmic derivative of the Evans function on a given closed curve and obtain the winding number of $E(\lambda)$ along that curve.  
In our case, the contour of integration is chosen so that it lies in the region defined by (\ref{ancondition}) and so that its interior encloses the  intersection of the   right side of the complex plane and the region defined by (\ref{imRinu7}). For example, in the case  $\eps=0.1$, $h=0.3$, $Z=6$, $\sigma=0.25$, $\delta=0.0005$, and $T_{ign}=0.01$
 we find the  RHS of (\ref{imRinu7}) to be $187.4478$ and $c$ to be $1.8588$. {\cre{For the computation of the Evans function, we find that we can take $L$ to be 2.7786 and}} our numerical winding number computation then shows that the Evans function has no zeroes other than the one at the origin. Note that the eigenvalue at $\lambda=0$ is due to the translation invariance of the system \eqref{r}. {\cre{We have made the same computations for several values of $h$ and $\sigma$ between 0 and 1 (specifically for $h=0.3, 0.6, 0.9$ and $\sigma=0.25, 0.5, 0.75$)}}, and  for $Z=4,5,6$ (with $\eps=0.1$, $\delta=0.0005$, and $T_{ign}=0.01$). Every time, we have found the eigenvalue at the origin to be the only one. This thus strongly suggests that there is a regime in which the front solution is spectrally stable,  {\cre{with the exception of  essential spectrum touching the imaginary axis.}}

Note that  {\cre{there is a special parameter regime where $F$ in System  \eqref{r}  is a function of one variable only and where System  \eqref{r} is related to a well-studied combustion model.}} \cre{As we explain below, this connection indicates that there  are parameter  regimes  when the wave in  \eqref{r} is unstable. Indeed,  when} $h=1$, $\sigma =0$, $\delta=T_{ign}=0$, the function $F$ in \eqref{Fdefrc} reduces to $F=e^{-1/u}$ and System \eqref{r}  can be scaled to {the well studied  case of the model (\ref{gasslessgeneral}). 
\cre{More precisely}, if one uses the change of variables  $\widetilde{u} = u/Z$,  $\widetilde{y}=y$,  $t= (Z/e^Z)\widetilde{t}$, and $x=\sqrt{Z/e^Z}\,\widetilde{x}$, {\cre{then system}} \eqref{r}, becomes 
\eq{ u_{t}&=u_{xx}+yF(u),\\ y_{t}&={\cre{\eps}} y_{xx}- Zy F(u).}{gassless0}}System (\ref{gassless0}) is a combustion model  {\cre{with large Lewis number $\Le=1/\eps$ \cite{B, Varas02,EvansWeberGroup, BGHL}}} for which there is an instability {\cre{occurring}} at about $Z=6.5$ for small $\eps$ \cite{Weber97,Ghazaryan2013, EvansWeberGroup}. This instability is  caused by a pair of conjugate eigenvalues crossing the imaginary axis to the right side of the complex plane  \cite{Ghazaryan2013,EvansWeberGroup}. {\cre{To obtain an example of instability in our model, we use the fact that we know an instability for System (\ref{gassless0}) above and the fact that System \eqref{r} is equivalent to  System (\ref{gassless0}) for a specific set of parameter values. We consider sets of parameters close the values $h=1$, $\sigma =0$, $\delta=T_{ign}=0$ for which \eqref{r} is equivalent to  System (\ref{gassless0}). Specifically, we consider the set of parameter values: $\eps=0.01$, $h=1$, $\sigma=0$, $\delta=0.0005$, and $T_{ign}=0.01$ in System \eqref{r}. In this case, we find a pair of eigenvalues $\lambda=\pm 0.1692\ri $ when $Z=6.578$. These two eigenvalues move to the right side of the complex plane when $Z$ is increased. If we increase the value of $\eps$ to 0.1, and keep the other parameters the same, the pair of eigenvalues occur for $Z=7.036$ at $\lambda=\pm 0.1088\ri$. Note that in order to locate an eigenvalue, we use the fact that (as mentioned before) the integral of the logarithmic derivative of an analytic function $f$ along a closed curve $\Gamma$ can be used to obtained the number of zeroes inside that curve together with the fact that the integral
$$
\frac{1}{2\pi \ri}\oint_\Gamma\frac{zf'}{f}dz
$$
gives the sum of the zeroes inside $\Gamma$. We numerically compute these integrals to locate zeroes of the Evans function in a given domain. Generally, this technique is known as the method of moments \cite{Ghazaryan2013,Bronski86,Overman86}. However, since in our case we are only looking for one root at a time only, the use of the higher moments (where $z$ is replaced by $z^n$) is not necessary.  
}}

\section{Convective instability. \label{cs}}

As we showed above there are parameter regimes when the linearization of the system \eqref{rm} around the front has no unstable discrete spectrum, but  has  essential spectrum that extends to  imaginary axes.  We  have checked in Section ~\ref{S:ess} that  this  spectrum can be moved to the left of the imaginary axis  by a  weight  $\rm e^{\alpha\xi}$, $\alpha > 0$, therefore, on the linear level, the instability of the front is convective,  i.e. perturbations to the front that are localized sufficiently close to its tail do not  influence the dynamics near the interface of the front. In the frame moving with the front the perturbations are transported towards $-\infty$ and are guaranteed  to grow no faster than the reciprocal of the exponential weight.
 In  this section we study whether this is the case on the nonlinear level as well.  More precisely,   we show that  Theorems 3.14 and 3.16 from \cite{Ghazaryan10} apply and describe their conclusions. 
 
  First, we check that the nonlinearity in \eqref{rm} satisfies the  conditions of the mentioned above  theorems.  Recall that the front $(\widehat{u}(\xi),\widehat{y}(\xi))$  satisfies boundary conditions  $(u,y) \to (1,0)$  as $\xi\to -\infty$, and  $(u,y) \to (0,1)$  as  $ \xi\to+\infty$.
 Theorems 3.14 and 3.16 from \cite{Ghazaryan10} are formulated for a system that supports a front that approaches the zero state at $-\infty$. Therefore we rewrite  \eqref{rm} in terms of a new variable $r=u-1$,
\begin{equation}
\begin{aligned}
r_t=&r_{\xi\xi} + cr_\xi + yF(h(r+1)+(1-h)(1-y)), \\  
y_t=&\eps y_{\xi\xi} + cy_\xi - yF(h(r+1)+(1-h)(1-y)).\label{rmq}
\end{aligned}
\end{equation}
This system has equilibria $(0,0)$ and $(-1,1)$.
We write  the reaction term  of \eqref{rmq} in vector notation that are used in \cite{Ghazaryan10}
$$
R(r,y)=\begin{pmatrix}     yF(h(r+1)+(1-h)(1-y))\\- yF(h(r+1)+(1-h)(1-y))
\end{pmatrix}. 
$$

For the results of \cite{Ghazaryan10} to hold, function $R$ should be in  $R\in C^3$, therefore here we do not consider the reaction term with a jump, but only the smooth version of it given by \eqref{Fdefrc}.  It is easy to see that $R(r,0)\equiv 0$ {\cre{for any $r$}}, so the linearization of the system \eqref{rmq} about the equilibrium  $(0,0)$ has a triangular structure  that implies that the evolution of the perturbation $\tilde y$ to the $y=0$ component is independent of the behavior of the perturbation $\tilde r$ to the $r=0$ component,
\begin{equation}
\begin{aligned}
\tilde r_t=&\tilde r_{\xi\xi} + c\tilde r_\xi +  F(1)\tilde y, \\
\tilde y_t=&\eps \tilde y_{\xi\xi} + c\tilde y_\xi - F(1)\tilde y .\label{rmqlin}
\end{aligned}
\end{equation}
Let us turn our attention to the operator  that is given by the right hand side of  \eqref{rmqlin},
$$L\begin{pmatrix} \tilde r\\\tilde y\end{pmatrix} = \begin{pmatrix}    \partial_{\xi\xi} +c\partial_{\xi} & F(1)\\0&\eps \partial_{\xi\xi} +c\partial_{\xi} -F(1) \end{pmatrix} \begin{pmatrix} \tilde r\\\tilde y\end{pmatrix}.$$
For the theorems that we {\cre{wish}} to apply, it is important that 
 the linear operator $\eps \partial_{\xi\xi} +c\partial_{\xi} -F(1) $ has  strictly negative spectrum,  and therefore the perturbation  $\tilde y$ decays in time exponentially, and the operator  $\partial_{\xi\xi} +c\partial_{\xi} $ generates a bounded semigroup on the space $L^2$.

The second set of conditions is concerned with the properties of the front and the linearization of \eqref{rmq} about it. We know from the geometric construction  \cite{GGJ} that the front we consider  converges to its rest states at exponential rates. We also know from Section \ref{S:ess} that the marginally stable essential spectrum is moved to the left of the imaginary axis by the weight $e^{\alpha\xi}$ for $0<\alpha<c/2$. We then choose $\alpha$  so small that each of the components of the derivative of the front  $(\widehat{u\,}^{\prime}(\xi),\widehat{y\,}^{\prime}(\xi))$ belongs to the  weighted Sobolev space $L_{\alpha}^2$ which {for us is a space of functions $f$  such that $f(\xi)e^{\alpha\xi} \in L^2$.   }

Thus all of  the hypothesis of theorems 3.14 and 3.16 from \cite{Ghazaryan10} hold.


 Theorems 3.14  and  Theorem 3.16 then imply  a result which is similar to  Theorems~6.1  and~6.2 in \cite{GLM}  which were obtained for a reduction of the original system corresponding to $\epsilon=0$ case.  We refer the reader to \cite{GLM} for the exact formulations  and physical interpretations and  here just give a description of the obtained for the  System \eqref{rm} results. Suppose that  initially the perturbations to the front  are small in both  regular norm  (either  $H^1$ or $BUC$) and the weighted norm, then in the weighted norm the perturbations to all components of the front decay  exponentially fast, in other words, the front is orbitally stable with asymptotic phase; in the regular norm, the perturbation to $y$ component decays exponentially in the norm without the weight, but the perturbation to the $u$ component stays only bounded. 
 Theorem 3.16 in \cite{Ghazaryan10}  implies that if the initial perturbation in addition are small in $L^1$-norm,  then the perturbation to {\cre{the $u$-component}} decays  diffusively  in $L_{\infty}$-norm.
 The similarity of the nonlinear stability results should not be surprising {cre{since}} the fronts considered  here and in \cite{GLM} correspond to different, singular reductions of the same  system.

\section{{{Discussion and conclusion}}}
\label{DC}

In this paper, we have considered Model (\ref{phys0}) describing combustion in inert porous media under condition of high hydraulic resistance. This model can be scaled to System (\ref{rr}) and we were interested in traveling front solutions, whose existence and unicity was established in  \cite{Dkh,GKS}. Assuming that the Lewis number has the special value ${{\Le^{-1}=\gamma(1-\mu)}}$, the system can be written as \eqref{rrnn}. Furthermore, assuming that the initial conditions (\ref{incond2}) hold, the system can be reduced to (\ref{r}) (first obtained in \cite{GR, Gordon_review}). While our study focussed on this reduced system, we show below that the results obtained in this article still hold if (\ref{incond2}) is not satisfied. Concerning the reduced system (\ref{r}}), we were  interested in the stability analysis of the traveling front solutions which  satisfy the boundary conditions (\ref{bc0}). We have considered the eigenvalue problem arising from the linearization of the equations about the front solutions. However, in order to have a well-defined linear operator, we have opted for a smooth version of the reaction term (\ref{Fdefrc}) that can be made as close as desired to the discontinuous version given in (\ref{Fdef}). 
In a weighted version of $L^2$, the essential spectrum was shown to lie in the open left side of the complex plane. As far as the point spectrum goes, we have performed an energy estimate computation to obtain an upper bound on the modulus of any possible eigenvalue with non-negative real part. Using this bound, we have used a numerical winding number computation of the Evans function to show that there are parameter regimes for which there {\cre{are}} no eigenvalues on the right side of the complex plane. Specifically, we have considered an array of values of $h$ and $\sigma$ between 0 and 1 {\cre{($h=0.3, 0.6, 0.9$ and $\sigma=0.25, 0.5, 0.75$)}} for $Z=4,5,6$, while keeping $\eps=0.1$, $\delta=0.0005$, and $T_{ign}=0.01$.  For the parameter regimes with  no unstable discrete spectrum, we further have proved that the instability caused by the essential spectrum is convective, which means perturbations to the front  that are localized sufficiently close to its tail do not  influence the dynamics near the interface of the front.

We would now like to show how our results  extend to the case where the initial condition \eqref{incond2} do not hold. Concerning the spectral stability, if we remove the initial condition given in \eqref{incond2} but still consider the value ${{\Le^{-1}=\gamma(1-\mu)}}$, then the system \eqref{rr}  becomes \eqref{rrnn} and thus differs from the 
reduced system \eqref{r} by the addition of the heat equation  $g_t=\eps g_{xx}$. 
The eigenvalue problem corresponding to the heat equation has no point spectrum and its continuous spectrum is restricted to  the open left side of the complex plane except for the point at the origin. As a consequence, the results obtained about the spectral stability obtained for the reduced system \eqref{r} still hold for  the system \eqref{rrnn} which is   full system \eqref{rr} (or \eqref{phys0}) in the special case where ${{\Le^{-1}=\gamma(1-\mu)}}$.
 For the nonlinear stability,  we perform {\cre{an}} analysis similar to that in Section \ref{cs}. It is easy to see  that {\cre{theorems}} 3.14 and 3.16 from \cite{Ghazaryan10}  are applicable   to {\cre{System}}  \eqref{rrnn} as well.   We conclude that the time evolution of perturbations to  {\cre{the}} $u$ and $y$ components are the same as in reduced system  \eqref{r}.  In case when  initial perturbations to the front  are small in both  regular   $H^1$ or $BUC$ norm  and in  the weighted norm, the additional,  $g$ component 
stays bounded in the norm without the weight and decays exponentially in the weighted norm.  Taking into account behavior of the perturbation{\cre{s}} to the $y$ component which decays exponentially in all norms, we can conclude that  perturbations to the original $v$-component  behave the same way as perturbations to $g$.
If, in addition,  initial perturbations  are also small in $L^1$-norm,  then the perturbation{\cre{s}} to $g$,  and therefore $v$-components  of the front  not only stay bounded but also decay algebraically  in $L_{\infty}$-norm.

We point out that the stability of the fronts in the full Model (\ref{phys0}) for values of the Lewis number other than the one considered here  is still an open question and will be a subject of our future research. As mentioned above, the cases of $\epsilon =0$ considered in our previous paper \cite{GLM} is singular and the results obtained for that case do not imply that similar results have to hold for the Model (\ref{phys0}) in general.

To conclude, in this paper we investigated the stability of fronts for a specific value of the Lewis number. To do so,  we have first considered one of  the known singular reductions of Model (\ref{phys0}). The singular reduction considered here is based on choosing a  special value of the Lewis number which is a parameter that is not present in the other singular perturbation of (\ref{phys0}) studied in  \cite{GLM} (when $\epsilon=0$). Although the analysis  follows the same sequence of  steps as in \cite{GLM} for stability analysis of fronts  such as spectral stability and the proof of the  nonlinear stability for marginally stable spectrum, from the technical point of view  it is  significantly different from the one performed in \cite{GLM}. This can be easily understood by the fact that setting $\epsilon=0$ in (\ref{phys0}) changes the nature of the system of PDEs since it reduces the order by 2.  The results of this article show that the fronts that the full model  (\ref{phys0}) exhibits for ${{\Le^{-1}=\gamma(1-\mu)}}$ are indeed physical because they are in some sense stable.


 


\vspace{.5cm}

 { \bf Funding statement.} 
This work was  supported by the National Science Foundation through grants DMS-1311313 (A. Ghazaryan) and DMS-0908074 (S. Lafortune).



\begin{figure}
\hspace{-0cm}
\scalebox{.73}{{\includegraphics{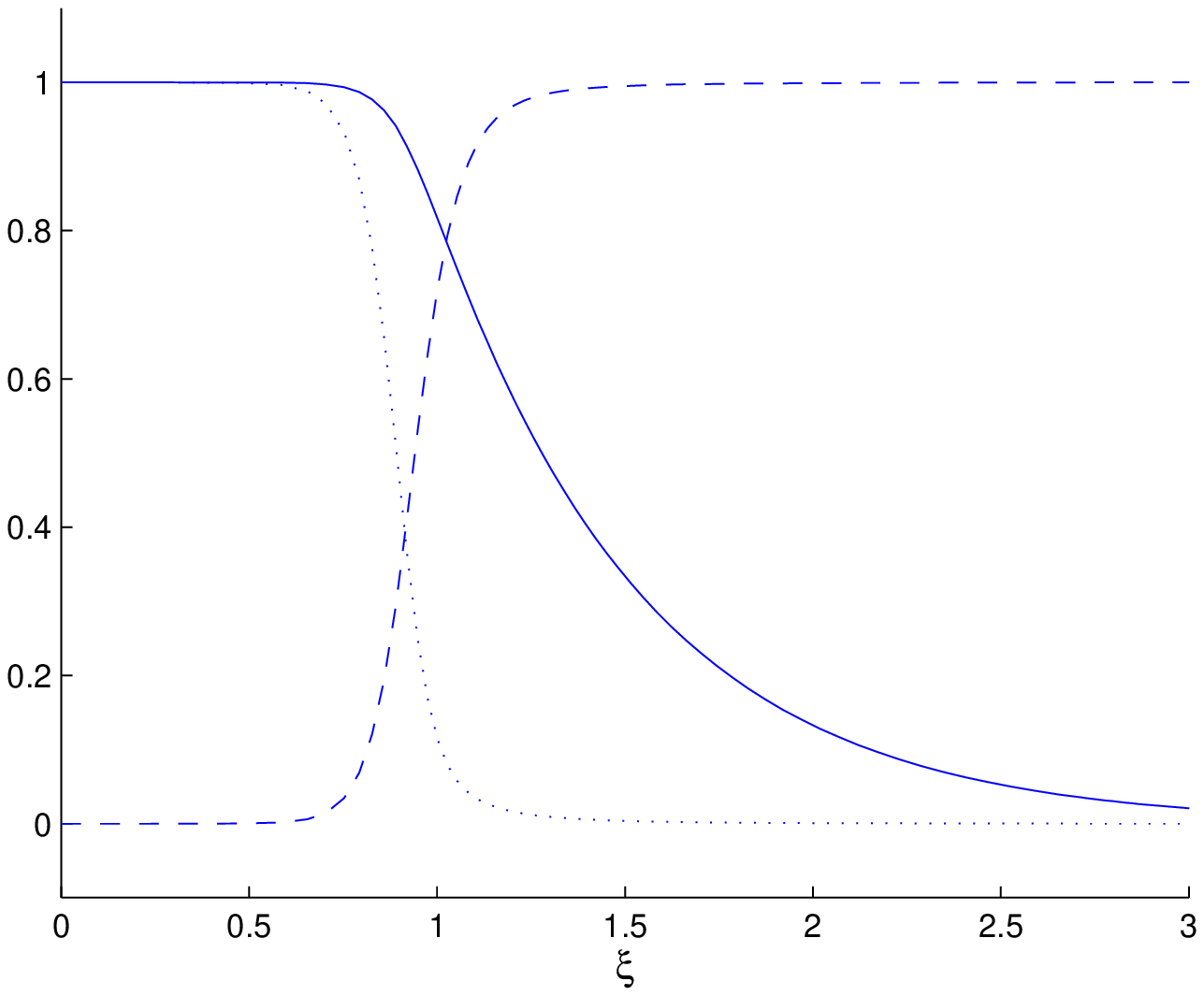}}}
\caption{\label{FigFront} Front solution of the system (\ref{e:z}) in the case $c=1.8588$, $\eps=0.1$, $h=0.3$, $\sigma=0.25$, $\delta=0.0005$, {and $T_{ign}=0.01$}, and $Z=6$. The solid line corresponds to $u$, the dashed line to $y$, and the dotted line to $z$.
}
\end{figure}

\begin{figure}
\subfigure[]{
{\includegraphics[width=200pt]{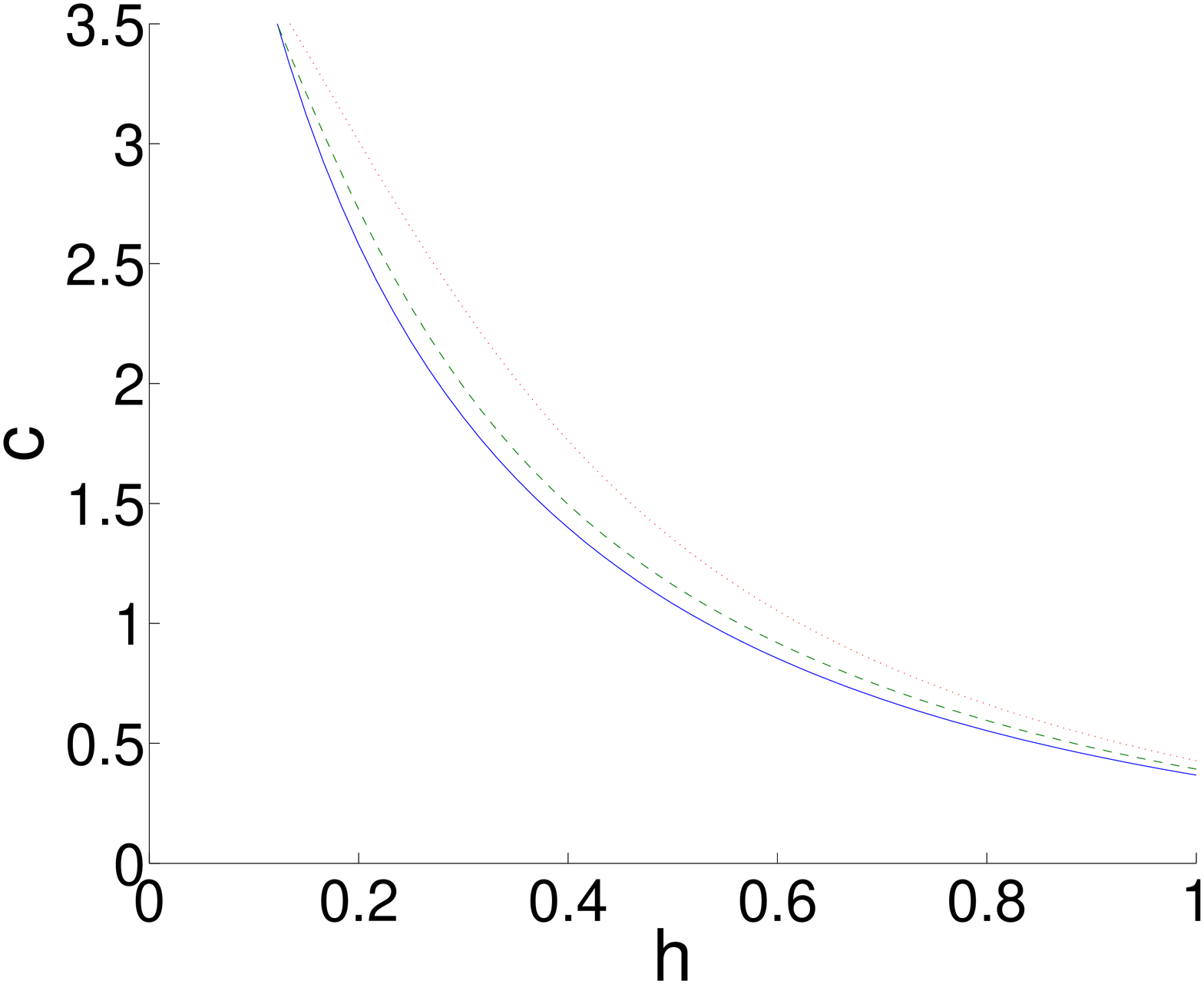}}  
\label{fig:subfig1}
}
\subfigure[]{
{\includegraphics[width=200pt]{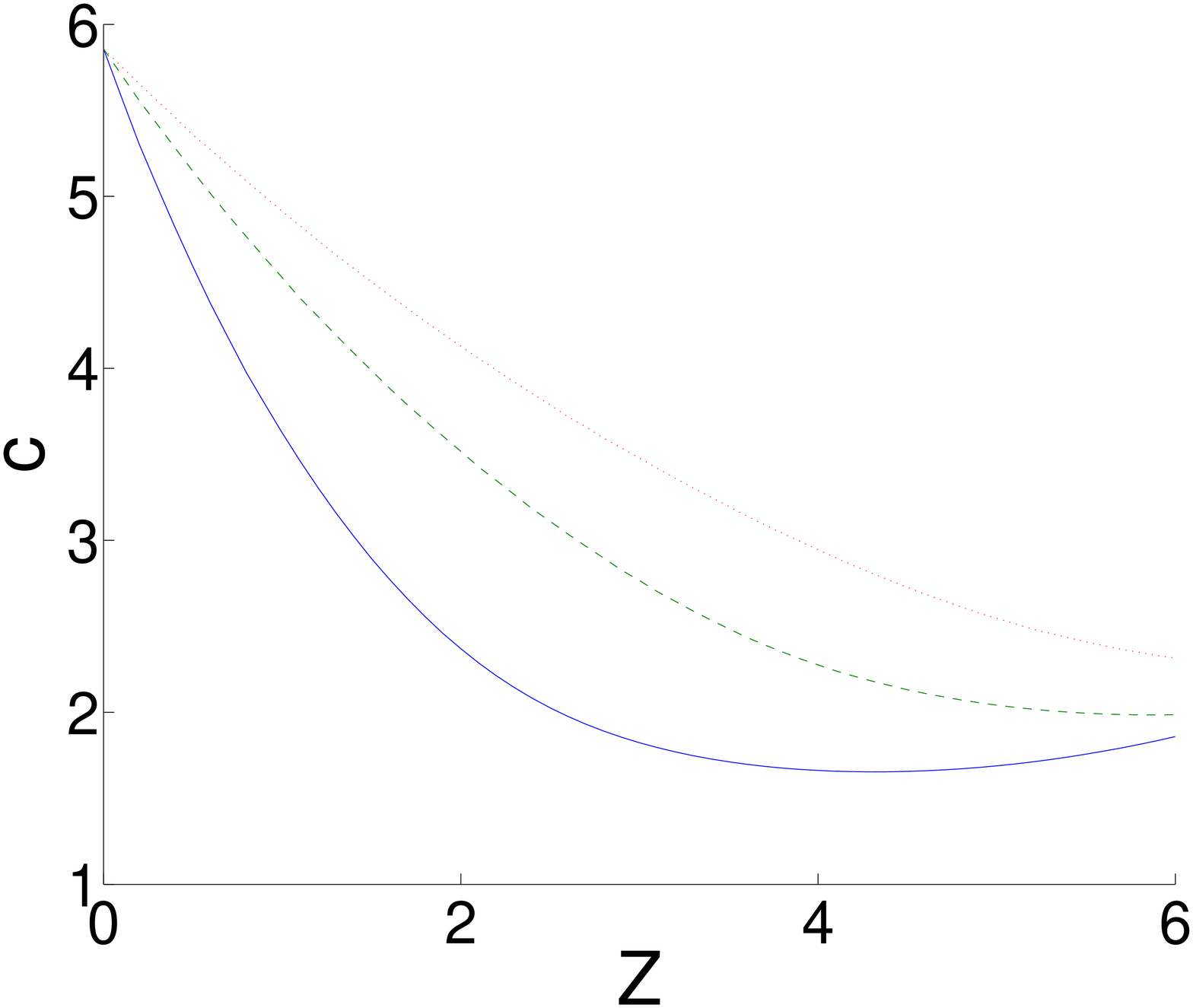}}  
\label{fig:subfigb}
}

\subfigure[]{
{\includegraphics[width=200pt]{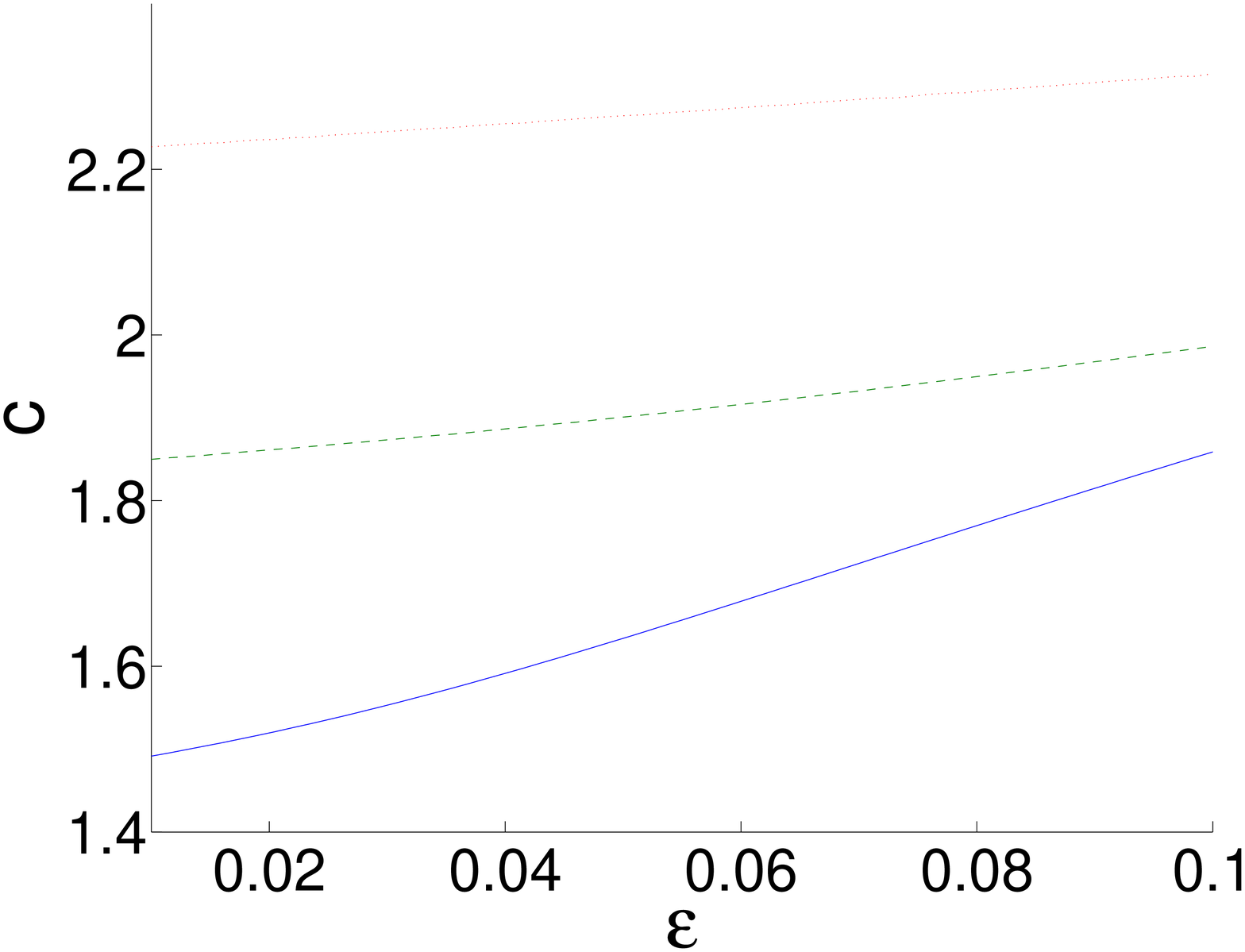}}  
\label{fig:subfigc}
}\subfigure[]{
{\includegraphics[width=200pt]{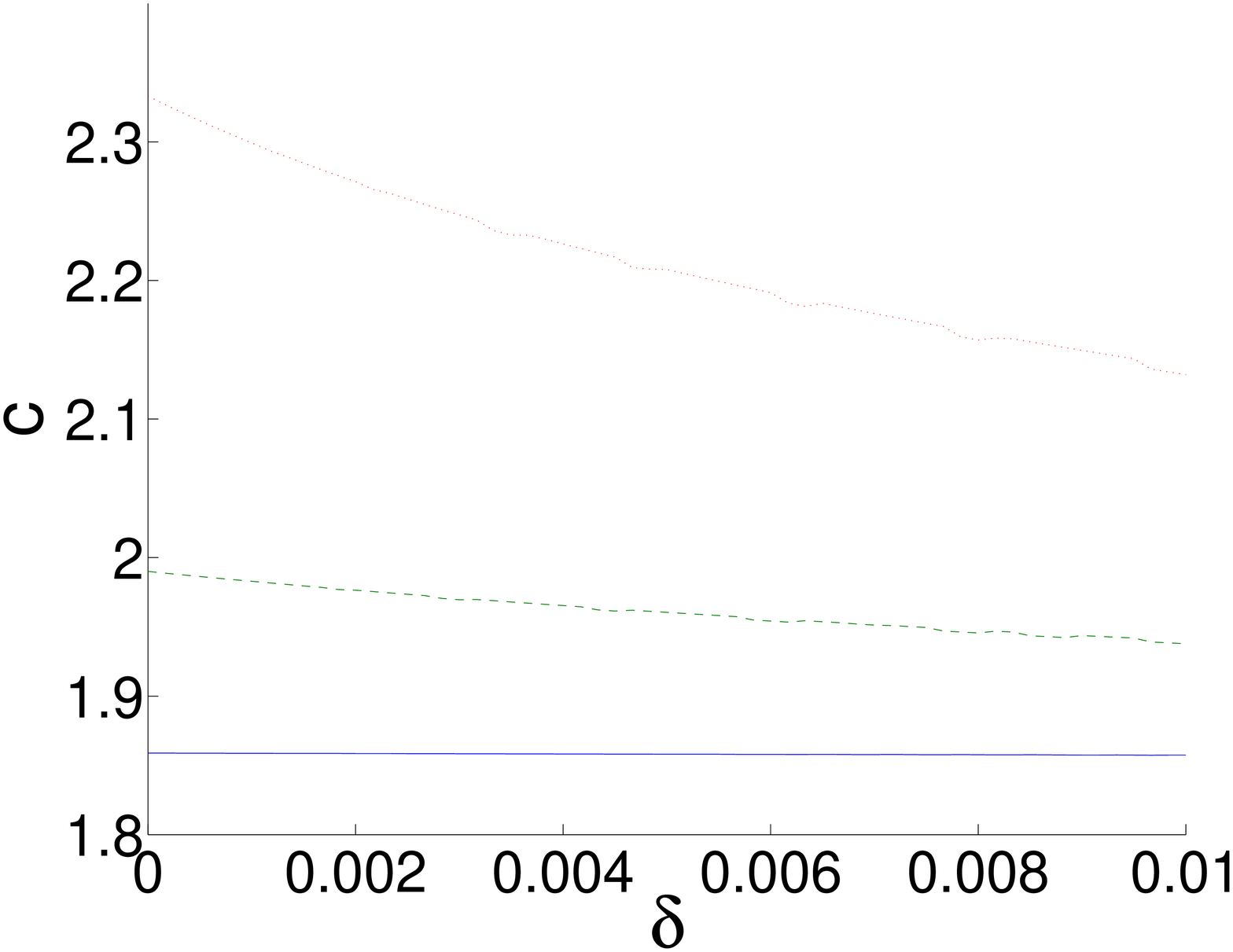}}  
\label{fig:subfigd}
}\caption{{{Plots showing the speed $c$ of the front as a function of the parameters (a) $h$ (with $\eps=0.1$, $Z=6$, $\delta=0.0005$, {and $T_{ign}=0.01$}), (b) $Z$ (with $\eps=0.1$, $h=0.3$, $\delta=0.0005$, {and $T_{ign}=0.01$}), (c) $\eps$ (with $h=0.3$, $Z=6$, $\delta=0.0005$, {and $T_{ign}=0.01$}), and (d) $\delta$ (with $\eps=0.1$, $h=0.3$, $Z=6$, and $T_{ign}=0.01$). {\cre{In each graph}}, the solid line corresponds to $\sigma=0.25$, the dashed line to $\sigma=0.5$, and the dotted line to $\sigma=0.75$. 
\label{SpeedFigure}
\vspace{-0.5cm}}}}
\end{figure}

\end{document}